\documentclass{article}
\usepackage[utf8]{inputenc}
\usepackage{lipsum}
\usepackage{url}
\usepackage[T1]{fontenc}
\usepackage[utf8]{inputenc}
\usepackage[english]{babel}
\usepackage{times}
\usepackage[left=2.3cm,top=3cm,right=2.3cm,bottom=3cm]{geometry}

\usepackage{blindtext}

\usepackage{amsmath}
\usepackage{amssymb}
\usepackage{amsthm}
\usepackage{hyperref}
\usepackage{graphicx}
\usepackage{subcaption}
\usepackage{caption}
\usepackage{xcolor}
\usepackage{longtable}
\usepackage{verbatimbox}

\usepackage[title]{appendix}
\usepackage{bm}
\usepackage{empheq}
\usepackage{placeins}
\hypersetup{
colorlinks,
linkcolor={red!50!black},
citecolor={blue!50!black},
urlcolor={blue!80!black}
}
\newcommand\scalemath[2]{\scalebox{#1}{\mbox{\ensuremath{\displaystyle #2}}}}
\theoremstyle{plain}
\newtheorem{theorem}{Theorem}[section]

\newtheorem{proposition}[theorem]{Proposition}

\newtheorem{corollary}[theorem]{Corollary}
\newtheorem{remark}[theorem]{Remark}
\newtheorem{definition}[theorem]{Definition}
\theoremstyle{definition}
\theoremstyle{remark}
\numberwithin{equation}{section}

\usepackage{subcaption}
\usepackage[small, sc]{titlesec}
\usepackage{filecontents}

\title{\textsc{Self-contact preventing energy for tubular rods}}

\author{ \textsc{C.\ Lonati}$^1$\thanks{\href{mailto:chiara.lonati01@universitadipavia.it}{\texttt{chiara.lonati01@universitadipavia.it}}}\,\,\,$-$\,\, \textsc{A.\  Marzocchi}$^2$\thanks{\href{mailto:alfredo.marzocchi@unicatt.it}{
			\texttt{alfredo.marzocchi@unicatt.it}}}
	\bigskip\\
	\normalsize$^1$ Dipartimento di Matematica ``F. Casorati'', Università di Pavia, via Ferrata 5, I-27100 Pavia, Italy.\\
	\normalsize$^2$ Dipartimento di Matematica e Fisica ``N. Tartaglia", Università Cattolica del Sacro Cuore,\\
	\normalsize via della Garzetta 48, I-25133 Brescia, Italy.\\
}

\date{}
\begin{document}
	
	\maketitle
	\begin{abstract}
		\noindent We introduce a generalization of M\"obius energy for knots to an energy functional for tubular neighbourhoods of closed inextensible curves. We prove the continuity of the energy and its boundedness for physically admissible tubes without self-contact and in particular for the torus. The functional allows to distinguish isotopy classes of the centerline through its physical inspiration to a self-repulsive electrostatic energy. If the tube has zero thickness, O'Hara's functional is recovered. Finally, a discussion on the possible exponents in the functional is carried out.
	\end{abstract}
	
	\noindent \textbf{Keywords} Self-contact, interpenetration of matter, M\" obius energy, tubular neighbourhoods\\
	\textbf{MSCcodes} 57K10, 70G75, 74K10, 74M15

	\section{Introduction}
	In the theory of knots \cite{adams1994knot,burde2003math,crowell2012introduction}, it is customary to introduce configurations avoiding self-intersection points, the so-called multiple points. This is both a requirement and a necessary condition to subdivide all possible configurations into disjoint classes, named ambient isotopy classes.\\
	For this reason, O'Hara \cite{o1991energy} defined an energy functional on a knot, called M\"obius energy, which prevents self-intersections distinguishing different types of knots and that was proved \cite{freedman1994mobius} to generate a minimum energy configuration where a sort of ``optimal distance with itself'' is obtained. Though not properly physical (the case of electrostatic energy was not contemplated in the theory), O'Hara ideas were very fruitful and led to many further results, like the fact that the circle actually reaches the absolute minimum and the proof of the existence of a minimum energy configuration of the knot in each isotopy class \cite{freedman1994mobius}, in addition to some generalizations of M\"obius energy for knots and the study of minimizers for its coupling with an elastic energy on the knot \cite{o1992family,kim1993torus,o1994energy,buck1995simple,simon1996energy,von1998minimizing,o2008energy,blatt2013stationary}.\\
	O'Hara's functional contains a term $1/|x-y|^\alpha$, where $|\cdot|$ is the Euclidean distance between points on the knot, which diverges as $x$ tends to $y$, subtracted by a term $1/d(x,y)^\alpha$ where $d$ is the distance between $x$ and $y$ on the knot, which diverges too as $x$ approaches $y$. Nevertheless, the difference is convergent as $x\to y$, so that it is possible to integrate it over all possible couples of points on the knot and, when there are no multiple points, get a finite integral. When the knot has multiple points, however, $|x-y|$ tends to zero but $d(x,y)$ does not, so that the integral diverges and hence a knot with multiple points cannot be a minimizer. The exponent $\alpha$ must belong to a precise interval in $\mathbb{R}$ to have a well-defined energy for knots. For $\alpha=1$ one obtains the classical Coulomb electrostatic potential, that however does not seem to provide minimizers for O'Hara's functional.\\
	The aim of this paper is to generalize some features of this theory to three-dimensional objects which are ``close'' to a knot, hopefully easier to use for real physical experiments and more suitable for applications, like modelling of proteins or DNA folding or the description of charged filaments. Indeed, several studies approach the notion of ``thickness'' of a knot: in \cite{gonzalez1999global,schuricht2004characterization} this concept is related to the notions of ideal knot and global radius of curvature and in \cite{von1999elastic} a small thickness is assigned to a knot in form of an obstacle, in order to minimize the curvature functional on separated isotopy classes. In \cite{kusnert1998distortion,litherland1999thickness}, the thickness is related, as in O'Hara's works, to the distortion and the so called rope-length of a knot. A tentative generalization of O'Hara's approach was carried out in \cite{rawdon2002mobius}; however, there the thickness of the knot is not considered in the formulation of the energy, but is used to find bounds for M\"obius energy. An interesting general approach for three-dimensional surfaces is also proposed in \cite{kusner1994mobius}; however, we want here to explicitly describe the functional taking into account the non-vanishing thickness of the elastic filament.\\
	In many problems involving deformable rods \cite{antman,Schuricht_2002,gonzalez2002global,schuricht2003euler,chamekh2009modeling} as well as those joint with minimal surfaces, the so-called Kirchhoff-Plateau problem \cite{giusteri2016instability,giusteri2017solution,de2019anisotropic,bevilacqua2018soap,bevilacqua2019soap,bevilacqua2020dimensional,BBLM22,bevilacqua2023effects}, an appropriate choice of the elastic energy term leads to the existence of a minimum energy configuration, but non-interpenetration of matter must be imposed implicitly on the space of all configurations. Many studies were carried out to express non-interpenetration and self-contact conditions for elastic loops \cite{lonati2023selfcontact}, both from a local point of view, \cite{chamekh2009modeling,mlika2018nitsche,chamekh2020frictional}, and from a global one where the formulation does not involve quantities described through local geometric quantities, as for example in \cite{Ciarlet_1987,gonzalez2002global,Schuricht_2002}.
	However, it is almost impossible to perform a priori numerical experiments on the solution, since the space of admissible configurations depends on possible points of contact, i.e. on the solution itself. A theory in which impossibility of compenetration is automatically necessary would be welcome at least in this sense. In particular, we will consider in this treatment tubular neighbourhoods of closed curves with fixed length, that are a suitable first approximation for inextensible, unshearable closed Kirchhoff rods, for which many approaches were used to treat self-contact and non-interpenetration.\\
	Three-dimensional objects like tubes of finite thickness have the advantage that overlapping of points appears before their midlines self-intersect, but are complicated to treat since the form of the boundary comes into play. Furthermore, there is no easy or apparent equivalent of O'Hara's distance $d(x,y)$ on the surface of the tube which blows up at the same order of the Euclidean distance. This equivalent quantity should of course tend to O'Hara distance as the cross-section of the tube shrinks to its midline, and should be reasonably computable in the easiest case of the torus.\\
	As in \cite{o1991energy}, where O'Hara partitioned the space of knots into isotopy classes separated by ``walls'' with infinite M\"obius energy, our goal is the possibility of distinguishing ambient isotopy classes of the centerline of the corresponding tubular neighbourhoods. Moreover, we would like to split them into configurations without self-contact and configurations with self-contact or interpenetration of matter, in order to detect the physically admissible ones. The main improvement from O'Hara's works and its extensions (\cite{o1991energy,o1992family,kim1993torus,o1994energy,freedman1994mobius,buck1995simple,simon1996energy,o2008energy,von1998minimizing,blatt2013stationary}) is indeed the introduction of the rod thickness.\\
	In this article we propose and discuss a function, denoted with $d^*(x,y)$, which meets these two requirements. The function is explicitly defined for a tubular neighbourhood of a closed knot and the functional is shown to be finite if and only if the tube does not present self-contact or interpenetration. Moreover, we compute its value for a torus and show some properties as boundedness from below, continuity, divergence to infinity only in the case of self-contact or not admissible configurations and we prove that, for a cross-section with vanishing thickness, we recover O'Hara's functional.\\
	A remarkable property that arises when dealing with tubular neighbourhoods is that the torsion of the midline plays a special role in the expression of $|x-y|$ and in the definition of $d^*$. So, differently with O'Hara's theory, also higher derivatives of the midline are necessary to define the functional. As above, the case $\alpha=1$ that corresponds to the  Coulomb electrostatic potential does not provide a divergent functional, but we notice that this is not a true electrostatic energy.\\
	The plan of the paper is as follows.
	In Section 2 we introduce the model through classical knot theory and we define our object. In Section 3 we highlight some geometrical properties of a generic tubular neighbourhood and we define the quantities necessary for the definition of the energy functional. In Section 4 we proceed with the definition of the functional and in Section 5 we prove the fundamental properties of the model, similar to the ones proved for O'Hara's functional in \cite{o1991energy,freedman1994mobius}. Finally, in Section 6 we show how one can easily recover the classical M\"obius energy and in Section 7 we carry out a discussion on the suitable exponent for the terms contained in the functional, finding a result analogous to the one contained in \cite{freedman1994mobius}.
	\section{Knot theory approach}
	We start recalling some basic definitions that will be useful in what follows.  For the standard definitions on knots, which are homeomorphisms $S^1\rightarrow \mathbb{E}^3$, where $\mathbb{E}^3$ is the usual affine space of points, we refer to \cite{adams1994knot,burde2003math,crowell2012introduction}.
	The main ingredient for a knot theory approach to the problem is the concept of equivalence between two knots $K_0$ and $K_1$: the ambient isotopy, which is stronger than  simple isotopy.\\ As explained in detail in \cite{burde2003math}, two knots can be isotopic although they are different with regard to their knottedness: they are indeed both homeomorphic to the unit circle, and so to each other, because  any region where knotting occurs can be contracted continuously to a point. An ambient isotopy, instead, takes into account also the neighbouring points of the knot in the ``movement'' from $K_0$ to $K_1$.\\
	An ambient isotopy is intuitively defined as a rearrangement of a knot in $\mathbb{R}^3$ without letting it pass through itself. Thus, as detailed in \cite{adams1994knot}, p.12, as a simple clarifying example, one is not allowed to shrink a part of the knot to a point. The following definition is contained in \cite{crowell2012introduction}.
	\begin{definition}
		\label{ho}
		Two knots $K_0$ and $K_1$ are \emph{ambient isotopic}
		if there exists an orientation-preserving continuous map $H:\mathbb{E}^3\times [0,1]\rightarrow \mathbb{E}^3$, such that for every fixed $t\in[0,1]$, the map $H_t:\mathbb{E}^3\rightarrow\mathbb{E}^3, \{
		x\mapsto H(x,t)\}$
		is a homeomorphism, with $H_0=\operatorname{Id}_{\mathbb{E}^3}$ and $H_1\circ K_0=K_1$ and with $x$ a point of the knot.
	\end{definition}
	The letter $t$ suggests that $[0,1]$ can be seen as a sort of ``time'' interval: indeed, the point $H_t(x)$ traces the path of the point $x$ during the ``motion'' of the knot from its initial position on $K_0$ to its final position on $K_1$.
	\begin{definition}
		A knot that is parametrized by an arc-length of class $C^1$ is called tame. Every tame knot is ambient isotopic to a polygonal knot.
	\end{definition}
	We will always consider only tame curves, because we will require an even higher regularity.
	\begin{definition}
		A singular knot is a smooth map $f:S^1 \rightarrow\mathbb{E}^3$ whose image has multiple points.
	\end{definition}
	Let $\mathcal{M}$ be the set of immersions of $S^1$ into $\mathbb{E}^3$ and $$\mathcal{K}=\{f:S^1\rightarrow \mathbb{E}^3: f \text{ is an homeomorphism}\}$$ the set of all knots. The set $\Sigma=\mathcal{M}\setminus\mathcal{K}$ is called \emph{discriminant set} and consists of all singular knots. It can be shown \cite{o2008energy} that two knots in $\mathcal{K}$ lie in the same ambient isotopy class if and only if they can be joined by a path in $\mathcal{K}$ that does not intersect $\Sigma$. Therefore, knot types are in one-to-one correspondence with the connected components, called \emph{cells}, of $\mathcal{M}/\Sigma$.  We will consider knots of given finite length $L$ and the arc-length describing the knot varying in the interval $[0,L]$. 
	In our treatment we will work with a tubular neighbourhood of the midline curve to approach the description of closed, inextensible and unshearable rods made in \cite{antman,Schuricht_2002}: we consider a closed curve $\gamma$ in $\mathbb{R}^3$ of fixed length $L$ parametrized by the arc-length parameter $s\in [0,L]$. We assume that this curve is of class $C^3([0,L];\mathbb{E}^3)$, so that the second derivative is well-defined and continuous as well as the standard Serret-Frenet frame $(\boldsymbol{t},\boldsymbol{n},\boldsymbol{b})$. The closure conditions read $\gamma(0)=\gamma(L),\,\boldsymbol{t}(0)=\boldsymbol{t}(L),\, \boldsymbol{n}(0)=\boldsymbol{n}(L) $, which obviously imply $\boldsymbol{b}(0)=\boldsymbol{b}(L)$.\\
	We will denote a closed circular tubular neighbourhood with radius $r>0$ of a curve $\gamma$ the set 
	\begin{equation}
		T_r[\gamma]=\{P\in \mathbb{E}^3: |P-Q|\leq r, Q\in \gamma\}
	\end{equation}
	where $|P-Q|$ is the Euclidean distance between two points in $\mathbb{E}^3$. In the sequel we will only write $T_r$.\\
	We denote with $p$ the mapping from a subset $\mathbb{R}^3$ to $\mathbb{E}^3$ that associates to $(\rho,s,\theta)$ the point $P$ corresponding to the coordinates $(\rho,s,\theta)$. 
	Thus, any point belonging to $T_r$ verifies 
	\begin{equation}\label{eq:Pdsth}
		p(\rho,s,\theta)-\gamma(s)=\rho\cos\theta\, \boldsymbol{n}(s)+\rho\sin\theta\, \boldsymbol{b}(s).
	\end{equation}
	where $s\in [0,L]$, $\rho\in [0,r]$ is the distance from the centerline in the plane determined by $\boldsymbol{n}(s)$ and $\boldsymbol{b}(s)$ and $\theta\in [0,2\pi]$ is the angle formed with the normal $\boldsymbol{n}(s)$. We notice that thanks to the regularity of $\gamma$ the mapping $p$ turns out to be continuous, surjective and invertible. We will use the notation $\Omega$ for the set $[0,+\infty[\times [0,L]\times [0,2\pi]$ and so $p(\Omega)$ will describe the shape of the tube in space.
	\\
	The boundary $\operatorname{bd} T_r$ of $T_r$ is the set of points where $r>0$ is fixed:
	\begin{equation}
		\operatorname{bd}T_r=\{P=\gamma(s)+r \cos \theta \boldsymbol{n}(s)+r \sin \theta \boldsymbol{b}(s), s\in[0,L], \theta\in [0,2\pi]\}
	\end{equation}
	while the interior part $\operatorname{int} T_r$ of $T_r$ is clearly the set of points 
	\begin{equation}\operatorname{int} T_r=\{P=\gamma(s)+\rho\cos\theta \boldsymbol{n}(s)+\rho\sin\theta \boldsymbol{b}(s),  \rho\in [0,r[, s\in[0,L], \theta\in [0,2\pi]\}
	\end{equation}
	A representation of the above definitions is shown in Figure \ref{tube}.\\
	\begin{figure}[htp]
		\centering
		\includegraphics[width=0.65\textwidth]{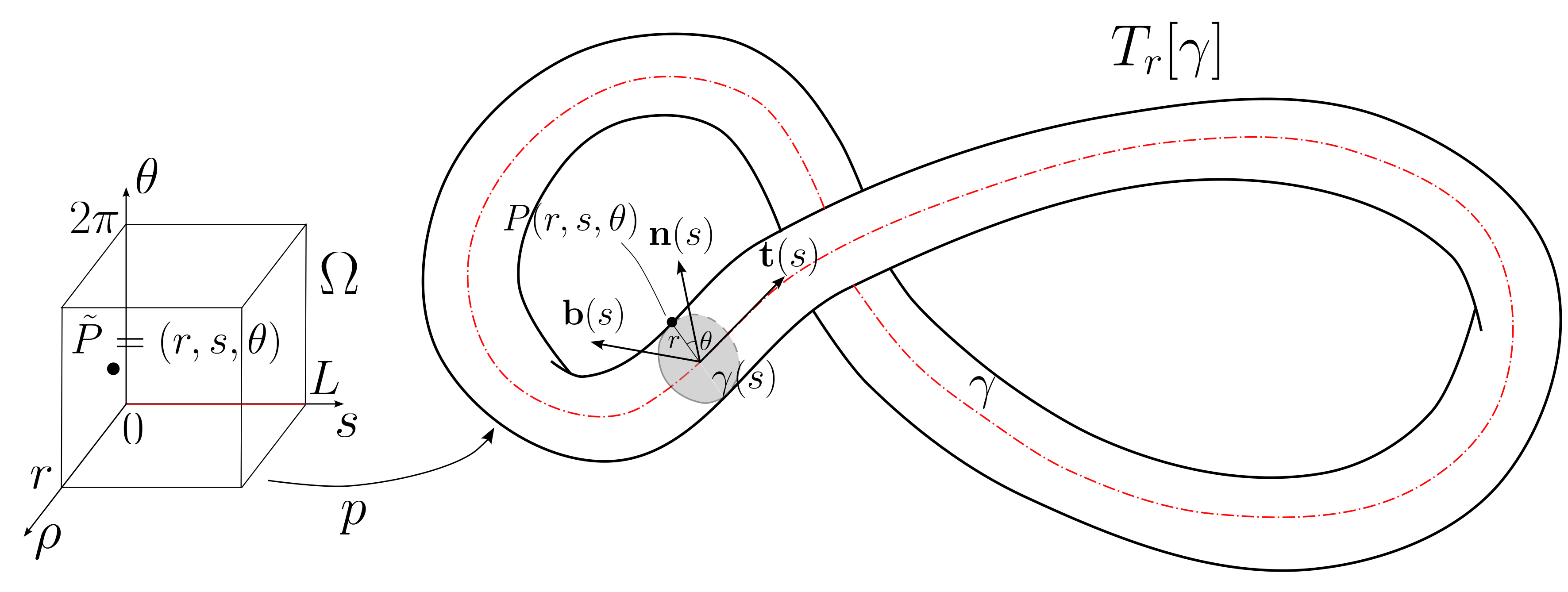}
		\caption{Representation of the map $p$ from the set $\Omega$ to $p(\Omega)$. A point $\Tilde{P}$ is mapped into $P\in p(\Omega)$, defined as above. In particular, here $P\in \operatorname{bd}T_r$.}
		\label{tube}
	\end{figure}
	\begin{definition}\label{def:interp1}
		A tubular neighbourhood of a curve $\gamma$ presents \emph{interpenetration of matter} if and only if there exist at least two distinct triples $x$ and $y$ in $\Omega$ with $p(x)=p(y)$ and at least one of the two points belongs to $\operatorname{int}T_r$. 
	\end{definition}
	\begin{definition}\label{def:selfcontact}
		A tubular neighbourhood of a curve $\gamma$ presents \emph{self-contact} if and only if there exist two distinct triples $x$ and $y$ in $\Omega$ with $p(x)=p(y)$ and $p(x),p(y)\in \operatorname{bd}T_r$. 
	\end{definition}
	\begin{remark}
	\emph{The boundary of the tubular neighbourhood is not in general the physical boundary of the object: this is true when there is no interpenetration. Indeed, when interpenetration occurs, the physical boundary is only a subset of the boundary of the tubular neighbourhood, because some parts of the boundary are contained into the interior part of $T_r$.}
	\end{remark}
	We recall now some known properties of tubular neighbourhoods.	\begin{proposition}\label{prop:tubular} If $\gamma$ is a curve of class $C^1$, then for every $r\geq0$, $T_r[\gamma]$ is the union, as $s$ ranges in $[0,L]$, of the closed spheres $B(\gamma(s),r)$ with centers $\gamma(s)$ and radius $r$ and also the union of all the full circles $\Gamma_r(s)$ obtained intersecating $B(\gamma(s),r)$ with the normal plane to $\gamma$ at $\gamma(s)$. If $\gamma$ has no singularities and $k=\max_{[0,L]}{\kappa(s)}<+\infty$ is the maximum curvature of $\gamma$, then there exists $r>0$ with $rk<1$ such that the boundary $\partial T_r[\gamma]$ of $T_r[\gamma]$ is the disjoint union of the boundaries $\partial\Gamma_r(s)$ of $\Gamma_r(s)$. In this case $\partial T_r[\gamma]$ is at least of class $C^1$, so it admits a tangent plane at every point. Conversely, if there exists $s\in [0,L]$ such that $r\kappa(s)>1$, then there exists $x\in\operatorname{int}\Omega$ and an open ball centered at $x$ such that $\operatorname{det}\nabla p(x)<0$ in that ball.    
	\end{proposition}
	The last part of the Proposition is developed in \cite{Schuricht_2002}.
	\\The situation where $r\kappa(s)>1$ will be called \emph{local interpenetration of matter}, that we want to avoid, so the configurations with $r\kappa(s)>1$ for an $s$ are considered not physically admissible. Clearly, it is an interpenetration of matter in the sense of Definition \ref{def:interp1}, but not the only possible one. Here, the  superposition of points is due to excessive bending and the points are in some sense near (that is, if the bending were a process, then there would be arbitrarily near points which were not superimposed before the critical bending and superimposed after it). We are more interested to rule out points far away on the boundary but which do come in superposition.\\
	If $T_r$ does not present interpenetration of matter, then $p$ is injective (this is sometimes called global injectivity, to distinguish it from the local conservation of orientation $\operatorname{det}\nabla p(x)>0$). We remark here that many conditions and formulations were studied to express the global injectivity of $p$, see \cite{lonati2023selfcontact} for details: while many conditions, mainly used for numerical examples, involve local expressions, i.e. local radius of curvature, or geometric characteristics expressed using the arc-length $s$ \cite{chamekh2009modeling,mlika2018nitsche,chamekh2020frictional}, some of them define global formulations on the object, as for example \cite{Ciarlet_1987,gonzalez2002global,Schuricht_2002}.\\
	The next result can be quite easily proved by continuity arguments. 
	\begin{theorem}
		\label{111}
		Suppose $H:\mathbb{E}^3\times[0,1]\to\mathbb{E}^3$ is a map that deforms continuously the points of a curve $\gamma$ as in Definition \ref{ho} and let $H$ be such that $H(\gamma,0)$ and $H(\gamma,1)$ belong to different ambient isotopy classes. Suppose moreover that $H(\gamma,t)$ verifies at every $t$ the conditions of Proposition \ref{prop:tubular} so that the tubular neighbourhood of $H(\gamma,t)$ of radius $r$ has always a regular boundary and it never has local interpenetration of matter. Then there exists $t^*\in]0,1[$ such that there is self-contact on $\operatorname{bd} T_r[H(\gamma,t^*)]$ but not for $t<t^*$. Furthermore, there exists $\tau\in]0,1[,\tau>t^*$, such that $T_r[H(\gamma,\tau)]$ has interpenetration of matter, and this does not happen for $t<t^*$.
	\end{theorem}
	Intuitively, if the midline $H(\gamma,t)$ has to change ambient isotopy class, then it must intersect itself somewhere, but then its tubular neighbourhood must come into contact before that time, and then $t^*$ is just the first time this happens. Obviously, shortly before $H(\gamma,t)$ becomes singular but after a first self-contact, its tubular neighbourhood must have a region of superposition.
	\begin{corollary}
		\label{inj}
		A mapping $p$, that describes a neighbourhood with interpenetration of matter, is not injective on the set $\operatorname{bd} T_r$, i.e. there is at least a point of self-contact on $\operatorname{bd} T_r$.
	\end{corollary}
	\begin{figure}[hbt]
		\centering
		\includegraphics[width=0.5\textwidth]{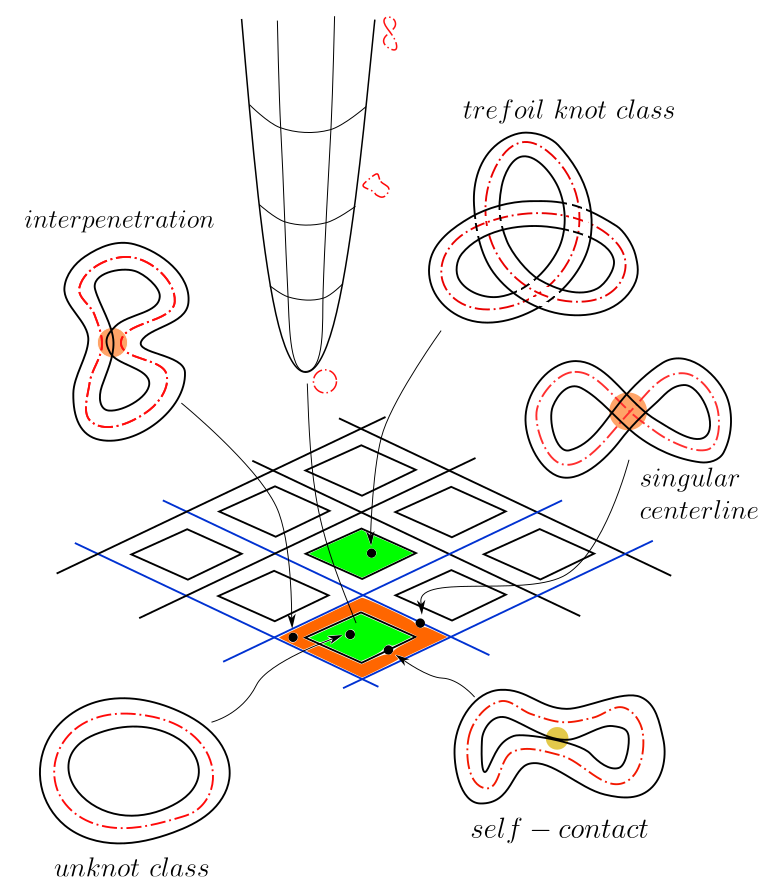}
		\caption{Representation of the desired subdivision: every tubular neighbourhood is identified with its centerline $\gamma$. The space of the possible knot types of $\gamma$ is partitioned into isotopy classes, separated by walls (blue lines) that represent configurations with singular knots as centerlines. Each class is divided into two regions: the green one, with physically admissible neighbourhoods and the orange one, with interpenetration of matter. The black squares represent the configurations with self-contact on the boundary and without interpenetration, so physically consistent. We show an example for each mentioned case in the class of the unknot, and an example of the trefoil knot. The paraboloid represents the ideal behaviour of the energy functional that we are looking for, analogous to the one proposed by O'Hara for knots, with a minimum in each class and the circular centerline as global minimizer.}
		\label{isotopy}
	\end{figure}
	These results allow us to use the same approach of O'Hara: given the radius $r$ of the neighbourhood, we partition the space of configurations of tubular neighbourhoods into cells
	defined by the knot type of the centerline; we want to separate these cells with walls of infinite energy, that consist of configurations in which the centerline is a singular knot; moreover, we would like to partition each cell into two regions: one region contains the centerlines that generate tubular neighbourhoods without interpenetration of matter, while the other region contains all the tubular neighbourhoods where interpenetration of matter occurs. Thanks to the previous results, the boundary between the two regions is then characterized by the presence of self-contact points on the boundary of the neighbourhood, but without interpenetration of matter. An intuitive representation of our ideal partition is presented in Figure \ref{isotopy}.\\
	We then need an energy functional that is infinite when there is self-contact or interpenetration of matter and is bounded otherwise. 
	\section{Geometrical setting of the problem}
	\begin{definition} Given the map $p$ defined by \eqref{eq:Pdsth}, we call the set $\mathcal{A}(s)=p(\cdot,s,\cdot)$ \emph{cross-section} at $s$, the curve $M(s)(\cdot)=p(r,s,\cdot)$ \emph{meridian} at $s$, and the curve $P(\theta)(\cdot)=p(r,\cdot,\theta)$ \emph{parallel} at angle $\theta$.
	\end{definition}
	By \eqref{eq:Pdsth}, then
	\begin{equation}
		\label{eq:defMerPar}
		\begin{aligned}
			M(s)(\theta)&=\gamma(s)+r \cos \theta \boldsymbol{n}(s)+r \sin \theta \boldsymbol{b}(s), \theta \in [0,2\pi]\\
			P(\theta)(s)&=\gamma(s)+r \cos \theta \boldsymbol{n}(s)+r\sin \theta \boldsymbol{b}(s), s\in[0,L]\
		\end{aligned}
	\end{equation}
	We remark that parallels and meridians are closed curves. In our case, having chosen a circular neighbourhood, the meridian is always a circle of radius $r$ with center at $\gamma(s)$ and lies in a plane perpendicular to $\boldsymbol{t}(s)$, $\forall s\in [0,L]$.
	Given two points $A(s,\theta)$ and $B(s,\varphi)$ belonging to the same meridian $M(s)$, we denote by $l_{M(s)}(A,B)$ their minimal distance along the meridian. 
	We have immediately
	\begin{equation}
		\label{meridian}
		\begin{aligned}
			l_{M(s)}(A,B):&=\begin{cases}
				r |\theta-\varphi|\qquad\qquad \text{if $|\theta-\varphi|\leq \pi$}\\
				r(||\theta-\varphi|-2\pi|)\quad \text{if $\pi<|\theta-\varphi|\leq 2\pi$}
			\end{cases}=\\
			&=r(2 h \pi+(-1)^h|\theta-\varphi|),\quad h=0,1,\,\,|\theta-\varphi|\in [h\pi,(h+1)\pi].
		\end{aligned}
	\end{equation}
	This is a positive and continuous function on $[0,2\pi]\times [0,2\pi]$. Moreover, $l_{M(s)}=l_{M(t)}$ for every $s,t \in[0,L]$.\\
	Similarly, given two points $A(s_1,\theta)$ and $B(s_2,\theta)$ belonging to the same parallel $P(\theta)$, we denote by $l_{P(\theta)}(A,B)$ their minimal distance along the parallel. Differentiating with respect to $s$ the expression of the parallel in \eqref{eq:defMerPar}, we get
	\begin{equation}
		\label{eq:lPth}
		\hat{l}_\theta(s_1,s_2):=\int_{s_1}^{s_2}\left\|\frac{\hfill d}{ds}P(\theta)(s)\right\|\,ds=\int_{s_1}^{s_2} \sqrt{(1-r \cos \theta \kappa(s))^2+r^2 \tau(s)^2}\, ds
	\end{equation} where $s$ does not cross $L$. Then, assuming $0\leq s\leq t\leq L$,
	\begin{equation}
		\label{parallel}
		\begin{aligned}
			l_{P(\theta)}(A,B)&=\begin{cases}
				\hat{l}_\theta(s,t) &\text{if $\hat{l}_\theta(s,t)\leq \Tilde{L}_\theta/2$}\\
				\hat{l}_\theta(t,L)+\hat{l}_\theta(0,s)=\Tilde{L}_\theta-\hat{l}_\theta(s,t)&\text{if $\Tilde{L}_\theta/2<\hat{l}_\theta(s,t)\leq \Tilde{L}_\theta $}
			\end{cases}=\\
			&=h \Tilde{L}_\theta+(-1)^h \hat{l}_\theta(s,t)\quad h=0,1,\,\,\hat{l}_\theta(s,t) \in\left[h \Tilde{L}_\theta,(h+1)\frac{\Tilde{L}_\theta}{2}\right].
		\end{aligned}
	\end{equation}
	This is a positive and continuous function of $s$ and $t$ on $[0,L]\times[0,L]$: indeed, it is the smallest length between the one obtained in clockwise direction and the counterclockwise one. From  \eqref{eq:lPth} we have that
	$$\left\|\frac{\hfill d}{ds}P(\theta)(s)\right\|^2= (1-r \cos \theta \kappa(s))^2+r^2 \tau(s)^2.$$
	While the first term is due to the distance along the parallel on $\operatorname{bd} T_r$  corresponding to the angle $\theta$ and projected along the vector $\boldsymbol{t}(s)$, we notice that the second term, which involves the torsion of $\gamma$, is orthogonal to the previous one and comes from the fact that the curve $\gamma$ leaves its osculating plane: indeed, for plane curves, this term vanishes. Anyway, although this distance is measured along meridians, it has no relationship with the distance $l_{M(s)}$, which has to be added in order to compute a sort of distance between generic points on $\operatorname{bd} T_r$.
	In view of this, we introduce the two finite lengths
	\begin{equation}
		\label{parallelt}
		l^{\boldsymbol{t}}_{P(\theta)}=\int_s^t ||P'(\xi)\cdot \boldsymbol{t}(\xi)||\,d\xi=\int_s^t |1-r \cos \theta \kappa(\xi)|\,d\xi=\int_s^t |1-r \cos \theta (\boldsymbol{n}(\xi)\cdot \boldsymbol{t}'(\xi))|\,d\xi
	\end{equation}
	i.e. the increment of the length of the parallel between the cross-sections $\mathcal{A}(s)$ and $\mathcal{A}(t)$ projected on the vector $\boldsymbol{t}(\xi)$, $\xi\in [s,t]$, and 
	\begin{equation}
		\label{parallelperp}
		\begin{aligned}l^{\boldsymbol{n},\boldsymbol{b}}_{P(\theta)}&=\int_s^t ||P'(\xi)\times \boldsymbol{t}(\xi)||\,d\xi=\int_s^t || -\tau(\xi)r\sin\theta \boldsymbol{n}(\xi)+\tau(\xi)r \cos\theta \boldsymbol{b}(s)||\,d\xi=\\
			&=r\int_s^t|\tau(\xi)|\, d\xi=r\int_s^t|\boldsymbol{n}(\xi)\cdot \boldsymbol{b}'(\xi)|\, d\xi
		\end{aligned}
	\end{equation}
	i.e. the increment of the length of the parallel projected on the boundary of $\mathcal{A}(\xi)$.\\ Although $l_{P(\theta)}^2\neq (l^{\boldsymbol{t}}_{P(\theta)})^2+(l^{\boldsymbol{n},\boldsymbol{b}}_{P(\theta)})^2$, this relationship is true ``infinitesimally'' and will help us to define the main functional  (see Figure \ref{elica} for the simple case of the helix).\\
	\begin{figure}[htb]
		\centering
		\includegraphics[width=0.583\textwidth]{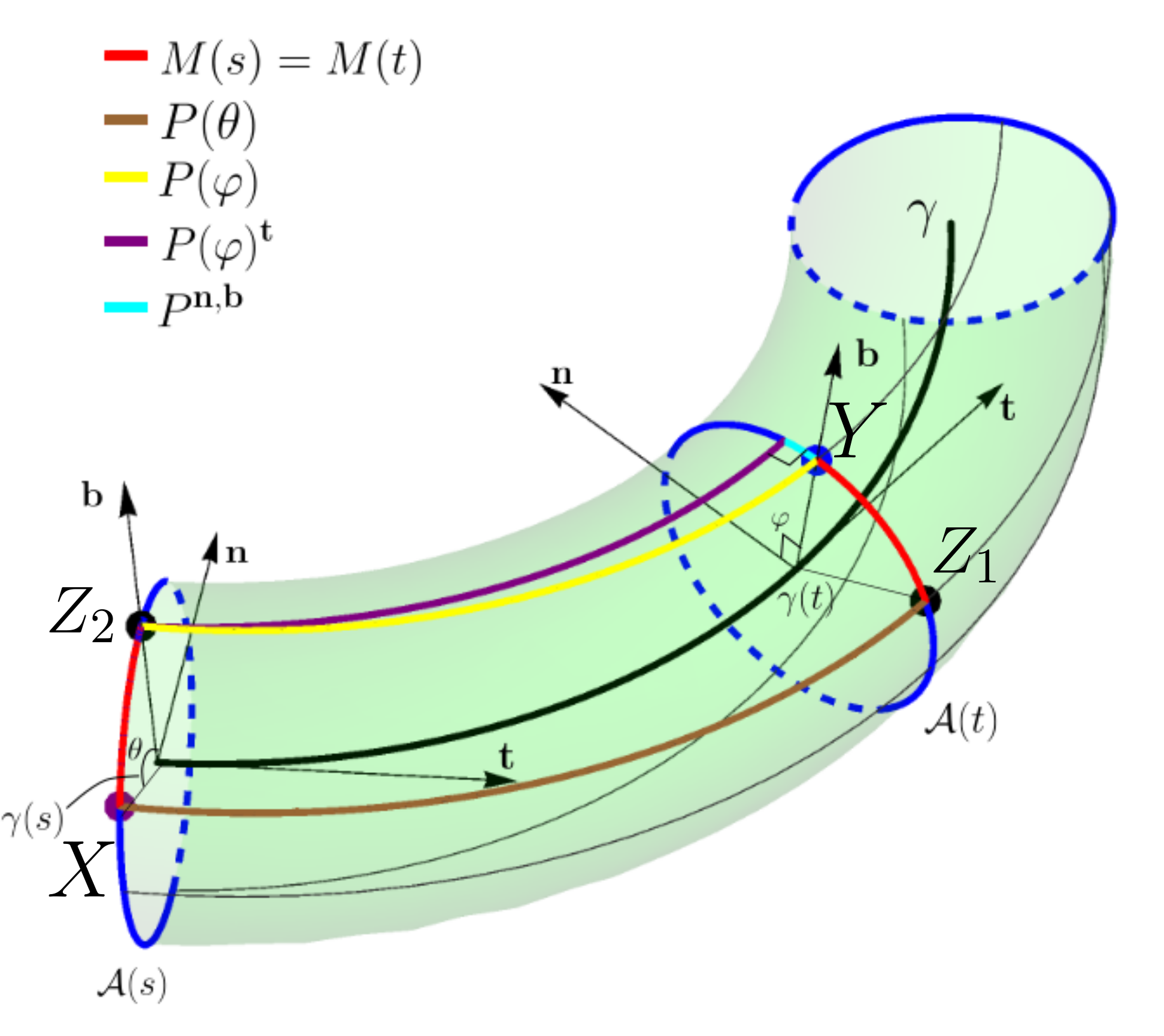}
		\caption{As an example, we plot an arc of helix described as the curve $\gamma$, with $\gamma(\xi)=(R\cos(\xi/\lambda),R\sin(\xi/\lambda),a\xi/\lambda)$, $\xi \in [0,7/2]$, $R=2,a=0.5,\lambda=\sqrt{R^2+a^2}$. We show the tubular neighbourhood $T_r$ around $\gamma$ with radius $r=0.5$, and we highlight the cross-sections corresponding to $s$ and $t$ and the Serret-Frenet frame in $\gamma(s)$ and $\gamma(t)$, with $s=0$ and $t=2$. Chosen the points $X(0,5\pi/6)$ and $Y(2,\pi/2)$, we have $Z_1(2,5\pi/6), Z_2(0,\pi/2)$ in the figure. We have the constant curvature $\kappa=R/\lambda^2\sim 0.47$ and the torsion $\tau=a/\lambda^2\sim 0.118$. In red, the two arcs of meridians and in yellow and brown the two parallels $P(\theta)$ and $P(\varphi)$. Finally, we show the decomposition of $P(\varphi)$ in a light blue arc $P^{\boldsymbol{n},\boldsymbol{b}}$ on the boundary of $\mathcal{A}(t)$ that has length $\hat{l}_P^{\boldsymbol{n},\boldsymbol{b}}=r \tau (t-s)$ and a purple curve $P(\varphi)^\mathbf{t}$ orthogonal to the cross-section $\mathcal{A}(t)$ with length $|t-s-r (t-s) \cos \varphi \kappa|$ that in our case equals $t-s$ because $\varphi=\pi/2$.}
		\label{elica}
	\end{figure}		
	\section{Definition of the functional}
	We recall the M\"obius energy functional introduced by O'Hara in \cite{o1991energy} for a knot $K$ described by a $C^2$ curve $\gamma$ with fixed length $L$, where $X$ and $Y$ are two points on the knot, whose position is described by their curvilinear abscissas in $ [0,L]$. The functional is 
	\begin{equation}
		\label{hara}
		\mathcal{E}(K)=\int_K\int_K \dfrac{1}{|X-Y|^2}-\dfrac{1}{{d}^2(X,Y)}dX\,dY=\int_0^L\int_0^L \dfrac{1}{|\gamma(s)-\gamma(t)|^2}-\dfrac{1}{(s-t)^2}ds\,dt
	\end{equation}
	and the integration is performed on the cartesian product between the knot and itself, where $d(X,Y)$ is simply the distance along the knot between the two points, i.e. the difference between their abscissas.\\
	We want to define a similar repulsive energy on a tubular neighbourhood in $\mathbb{E}^3$. We remark, however, that the exponent chosen for the distance between two charges in the electrostatic potential is $-2$, instead of the $-1$ of the Coulomb law. We first choose this exponent following O'Hara approach, but a more detailed discussion about it will be carried out below.\\
	As was done for the M\"obius energy, the main ingredient to distinguish different configurations is to make a comparison between the Euclidean distance between two points $X$ and $Y$ and an expression $d(X,Y)$ that takes into account the geometry of the configuration. In this way, when two generic points occupy different positions in $\mathbb{E}^3$ the functional is well-defined. Otherwise, the interesting cases are the following:
	\begin{itemize}
		\item[a)] $X=\gamma(s)=\gamma(t)=Y$ with $s=t$, i.e. $X,Y$ occupy the same position in space and they are determined by the same abscissa: then, as $X\to Y$, $|X-Y|$ and $d(X,Y)$ go to zero and the difference between their squared reciprocals may be finite or infinite (of course, one should have a bounded limit for the integral and this actually happens for the above distance $|s-t|$);
		\item[b)] $X=\gamma(s)=\gamma(t)=Y$ with $s\neq t$, i.e. $X,Y$ occupy the same position in space but come from different points on the abstract knot: then, the second fraction of the integrand is bounded, while the first term diverges to $+\infty$ and so does the energy: it is the case of a singular knot.
	\end{itemize}
	Minimizing the integral, b) is avoided, provided the integral is convergent in case a).
	
	We now want to generalize the previous functional to tubular neighbourhoods distinguishing isotopy classes and tubes with or without interpenetration of matter. Therefore, it is natural to use as a fact of distinction between these classes the presence of self-contact between points on the boundary of the tube.\\ This choice is also motivated by the fact that,  for the electrostatic inspiration of this energy,
	when we imagine our elastic material as a conductor charged in electrostatic equilibrium, the charge distributes on the external surface of the conductor and being the charge homogeneous, the object presents a repulsive character with itself.\\
	
	\noindent Given two points $X(s,\theta)$ and $Y(t,\varphi)$ on $\operatorname{bd} T_r$, we denote 
	\begin{equation}
		\label{eq:z1z2}
		Z_1(t,\theta)=P(\theta)\cap M(t),\quad Z_2(s,\varphi)=P(\varphi)\cap M(s) 
	\end{equation}
	i.e. the point belonging to the same meridian of $Y$ and the same parallel of $X$ and the point belonging to the same meridian of $X$ and the same parallel of $Y$, respectively. Next, we set $\hat{l}_{P(\theta)}$ the length of the parallel $P(\theta)$ between $X$ and $Z_1$ as in \eqref{parallel}, that will be decomposed into two components as before and analogously with the length $\hat{l}_{P(\varphi)}$ of the parallel $P(\varphi)$ between $Y$ and $Z_2$. Finally, by $l_{M(s)}$ we denote the length of the meridians corresponding to $s$ and $t$ between $X$ and $Z_2$ and $Y$ and $Z_1$ as in \eqref{meridian}.
	\begin{definition}
		Given a curve $\gamma$ of class $C^3$ and a tubular neighbourhood $T_r[\gamma]$ around $\gamma$ with radius $r>0$, we define
		\begin{equation}
			\label{dio}
			\begin{aligned}
				&F(T_r)=\int\limits_{\operatorname{bd}T_r \times \operatorname{bd} T_r}\left(\dfrac{1}{|X-Y|^2}-\dfrac{1}{{d^*}^2(X,Y)}\right)dS\,dS=\\ &=\int_0^L\int_0^L\int_0^{2\pi}\int_0^{2\pi}\left(\dfrac{1}{|p(r,s,\theta)-p(r,t,\varphi)|^2}-\dfrac{1}{{d^*}^2(p(r,s,\theta),p(r,t,\varphi))}\right)\,d\varphi\,d\theta\,dt\,ds
			\end{aligned}    
		\end{equation}
		where $|X-Y|$ is the Euclidean distance
		\begin{equation}
			\label{euclidean}
			|X-Y|(s,\theta,t,\varphi)=|\gamma(s)-\gamma(t)+ r(\cos \theta \boldsymbol{n}(s)- \cos \varphi \boldsymbol{n}(t))+r (\sin \theta \boldsymbol{b}(s)-\sin \varphi \boldsymbol{b}(t))|
		\end{equation}
		and ${d^*}^2(X,Y)$ is, assuming $s<t$ for simplicity, 
		\begin{equation}
			\label{dstar}
			{d^*}^2(X,Y)(s,\theta,t,\varphi)=(l_{M(s)}+\hat{l}^{\boldsymbol{n},\boldsymbol{b}}_{P})^2+\hat{l}^{\boldsymbol{t}}_{P(\theta)}\hat{l}^{\boldsymbol{t}}_{P(\varphi)}
		\end{equation}
		where the quantities are defined in \eqref{meridian}, \eqref{parallelt}, \eqref{parallelperp}.
	\end{definition}
	In words, we sum the projection of the parallel on the cross-section with the arc of meridian, so the first part is an arc on the boundary of the cross-section, while the second part is made up of two curves orthogonal to $\mathcal{A}(\xi)$, $\forall \xi \in[s,t]$.
	\begin{remark}
		\label{planar}
		\emph{If the midline $\gamma$ lies in a plane, then ${d^*}^2(X,Y)$ reduces to $\hat{l}_{P(\theta)}\hat{l}_{P(\varphi)}+l_{M(s)}l_{M(t)}=\hat{l}_{P(\theta)}\hat{l}_{P(\varphi)}+l_{M(s)}^2$, where $\hat{l}_P$ is defined in \eqref{parallel}, since the torsion vanishes, and the parallels are plane curves.}
	\end{remark}
	\begin{remark}
		\label{notadistance}
		\emph{${d^*}^2(X,Y)$ is a positive quantity which vanishes if and only if $X\equiv Y$ and is a symmetric function of $X$ and $Y$. However, $d^*$ is not a distance on $\operatorname{bd}T_r$ because it does not satisfy the triangle inequality.
		As a counterexample, one may take $X=(R-r,0,0)$, $Y=(0,R,r)$, $Z=\left(\left(R-r \cos \frac{\pi}{6}\right)\cos \frac{\pi}{3},\left(R-r \cos \frac{\pi}{6}\right)\sin \frac{\pi}{3},r\sin\frac{\pi}{6}\right)$, that lie on the boundary of a torus, and verify directly that
		$${d^*}^2(X,Y)\! \geq \! {d^*}^2(X,Z)+{d^*}^2(Z,Y)+2\sqrt{{d^*}^2(X,Z)}\sqrt{{d^*}^2(Z,Y)}\!=\!(d^*(X,Z)+d^*(Z,Y))^2$$}
	\end{remark}
	\begin{remark}                                         	\emph{An intuitive expression for playing the role of $d(X,Y)$ in our model would be the shortest distance on $\operatorname{bd} T_r$ between $X$ and $Y$, i.e. the length of the shortest geodesic on the surface between the two points. However, the difficulty in its explicit calculation is prohibitive for specific results on a generic configuration.}
	\end{remark}
	\section{Main results}
	We now prove the main results on our functional. Similarly to the knot energy \eqref{hara}, we can prove that it is bounded on a torus with circular centerline (instead that on a circular knot), that diverges only when two points on the boundary are different on the object (have different coordinates) but coincide in the space, that is bounded from below and continuous in the topology of the space of functions to which $\gamma$ belongs. The analogous results were proved in \cite{o1991energy} for M\"obius energy. However, we cannot prove as was done for \eqref{hara} in \cite{freedman1994mobius} that the functional is positive for every $T_r$, while this statement would be easily true using as $d^*$ the geodesic distance on $\operatorname{bd} T_r$.
	Indeed, the integrand of our functional can not have in general a definite sign, but this is ininfluent from the point of view of the Calculus of Variations, for instance.
	\begin{remark}
		\emph{The functional $F(T_r)$ is invariant under a reparametrization of $\gamma$. Indeed, it is sufficient to multiply the integrand of the functional by $|\gamma'(s)||\gamma'(t)|$, where in general it could be $|\gamma'|\neq 1$.}
	\end{remark}
	\begin{remark}
		\emph{Using minimal lengths, the integrand of the functional is continuous on the set of couples $(\theta,\varphi)$, such that $|\theta-\varphi|=\pi$ and on the set of couples $(s,t)$, such that $\int_s^t \sqrt{(1-r \cos \theta \kappa(u))^2+r^2 \tau(u)^2}\, du= \Tilde{L}_\theta/2$. We remark that on these set of points, it is not differentiable.}
	\end{remark}
	\begin{theorem}
		\label{thtoro}
		The functional $F(T_r)$ is positive and bounded on a tubular neighbourhood of a  circular centerline, i.e. a torus, that does not present interpenetration of matter in its interior or self-contact on its boundary.
	\end{theorem}
	\begin{proof}
		Let $R$ be the radius of the centerline of the torus in the plane $z=0$ and $r$ the radius of its generating circle; we assume $R>r$ to avoid interpenetration of matter or self-contact and we reparametrize the midline so that the new parameters $u,v$ vary from $0$ to $2\pi$. Clearly, it is sufficient to study the integrand of $F(T_r)$ with this parametrization. We start proving that the functional is bounded from above.\\
		We consider two generic points $X,Y$ on $\operatorname{bd}T_r$ with latitudes $\theta,\varphi$ and longitudes $u,v$. Setting 
		$$a=R-r\cos\theta,\ b=R-r\cos\varphi,$$
		we have $a,b>0$ since $R>r$ and then
		$$X=(a\cos u,a\sin u, r \sin \theta)\qquad Y=(b\cos v,b\sin v, r \sin \varphi).$$
		Moreover, setting $c=r(\sin \theta-\sin \varphi)$, the euclidean distance between $X$ and $Y$ is easily found as
		$$|X-Y|^2=a^2+b^2+c^2-2ab\cos(v-u).$$
		From \eqref{dstar}, ${d^*}^2(X,Y)$ is given by $r^2|\theta-\varphi|^2+ab|v-u|^2$, if $ |\theta-\varphi|\leq \pi$ and $|v-u|\leq \pi$ and 
		$$
		r^2(2 h \pi+(-1)^h |\theta-\varphi|)^2+a b(2 k \pi+(-1)^k |v-u|)^2$$
		in the other cases, for $h,k=0,1, |\theta-\varphi|\in [h\pi,(h+1)\pi], |v-u|\in [k\pi,(k+1)\pi]$.\\
		We perform the integration in the first case; the others are similar. 
		Set $x=v-u$ and suppose $v>u$ by symmetry, so that
		$$
		\begin{aligned}
			&F(T_r)=\int_{[0,2\pi]^4}\left(\frac{1}{|X-Y|^2}-\frac{1}{{d^*}^2(X,Y)}\right)\,dv\,du\,d\varphi\,d\theta=\\
			&=2\int_0^{2\pi}\int_0^{2\pi}\int_0^{2\pi}\int_0^\pi\left(\frac{1}{|X-Y|^2}-\frac{1}{{d^*}^2(X,Y)}\right)\,dx\,du\,d\varphi\,d\theta=\\
			&=4\pi\int_0^{2\pi}\int_0^{2\pi}\int_0^\pi\left(\frac{1}{|X-Y|^2}-\frac{1}{{d^*}^2(X,Y)}\right)\,dx\,d\varphi\,d\theta.
		\end{aligned}
		$$	
		
		Now we carry out the first integration and show that the function obtained is bounded from above in the set $[0,2\pi]\times [0,2\pi]$, with $ |\theta-\varphi|\leq \pi$.\\
		First, we have
		$$\int_0^\pi\dfrac{1}{|X-Y|^2}\,dx=\int_0^\pi \dfrac{1}{a^2+b^2+c^2-2ab\cos x}\,dx=\frac{\pi}{\sqrt{((a-b)^2+c^2)((a+b)^2+c^2)}}$$
		Since $a,b>0$, this integral is unbounded if and only if $a=b$ and $c=0$, i.e. $\theta=\varphi$. 
		On the other hand, we have 
		$$
		\int_0^\pi \frac{1}{{d^*}^2(X,Y)}\,dx=\int_0^\pi \frac{1}{r^2|\theta-\varphi|^2+abx^2}\,dx=\frac{1}{r|\theta-\varphi|\sqrt{ab}}\arctan\left(\frac{\sqrt{ab}\pi}{r|\theta-\varphi|}\right).
		$$
		So, for $\theta\neq \varphi$ both integrals are bounded from above and so is the functional because in the next integrations $\theta,\varphi$ are bounded.  Then we analyze the limit of the integrand when $\varphi\to\theta$. For $\varphi \to \theta$ we have $b\to a>0$, and using the expansion
		$$
		\arctan\left(\frac{\sqrt{ab}\pi}{r|\theta-\varphi|}\right)\sim \frac{\pi}{2}-\frac{r|\theta-\varphi|}{\sqrt{ab}\pi}
		$$
		the limit becomes, naming now $y=\theta-\varphi$, and after some calculations,
		\begin{multline}
			\label{integrando}
			\lim_{y\to 0}\!\left(\!\frac{\pi}{4r\left|\sin\left(\dfrac{y}{2}\right)\right|\sqrt{\! \left(\! R-r \cos\left(\! \theta+\dfrac{y}{2}\right)\left(\! \cos\left(\dfrac{y}{2}\right)\! \right)\! \right)^2 \! +2r^2\left(\! \cos\left(\theta+\dfrac{y}{2}\right)\sin\left(\dfrac{y}{2}\right)\! \right)^2}}\right.\\
			-\left.\frac{\pi}{2r|y|\sqrt{ab}}+\frac{1}{\pi ab}\right)
			=\frac{1}{\pi a^2}<+\infty.
		\end{multline}
		We notice that the part containing $y$ is infinitesimal since 
		$$\lim_{y\to0}\left(\frac{1}{\sin y}-\frac{1}{y}\right)\sim\frac{y}{6}$$
		and this shows that the squared Euclidean distance and ${d^*}^2$ degenerate at the same order when $X\to Y$, so the two divergences compensate.\\
		\begin{figure}[htb]
			\centering
			\subfloat[]{\label{trapeziotoro}\includegraphics[width=0.48\textwidth]{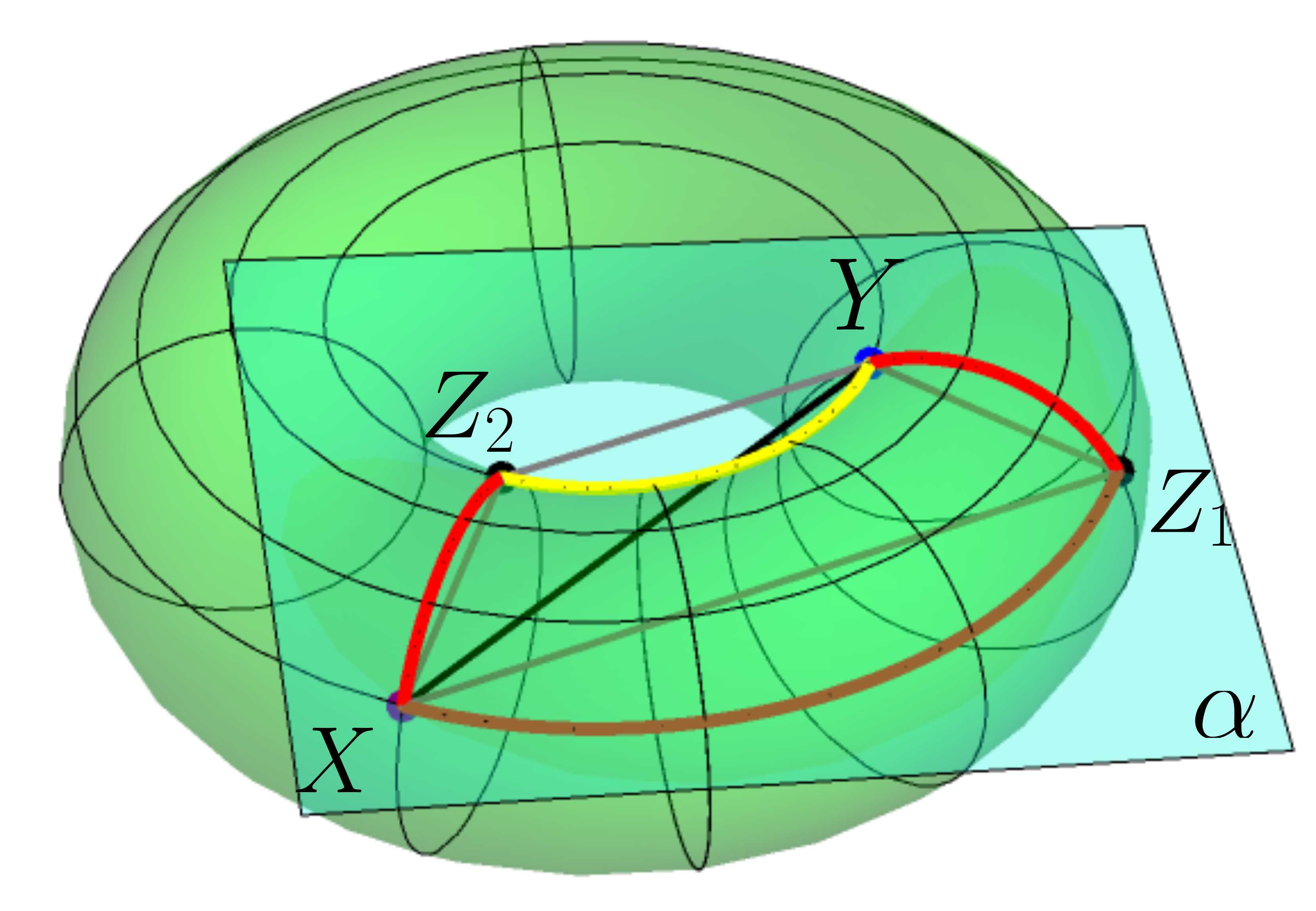}}
			\subfloat[]{\label{trapeziodasoloo}\includegraphics[width=0.48\textwidth]{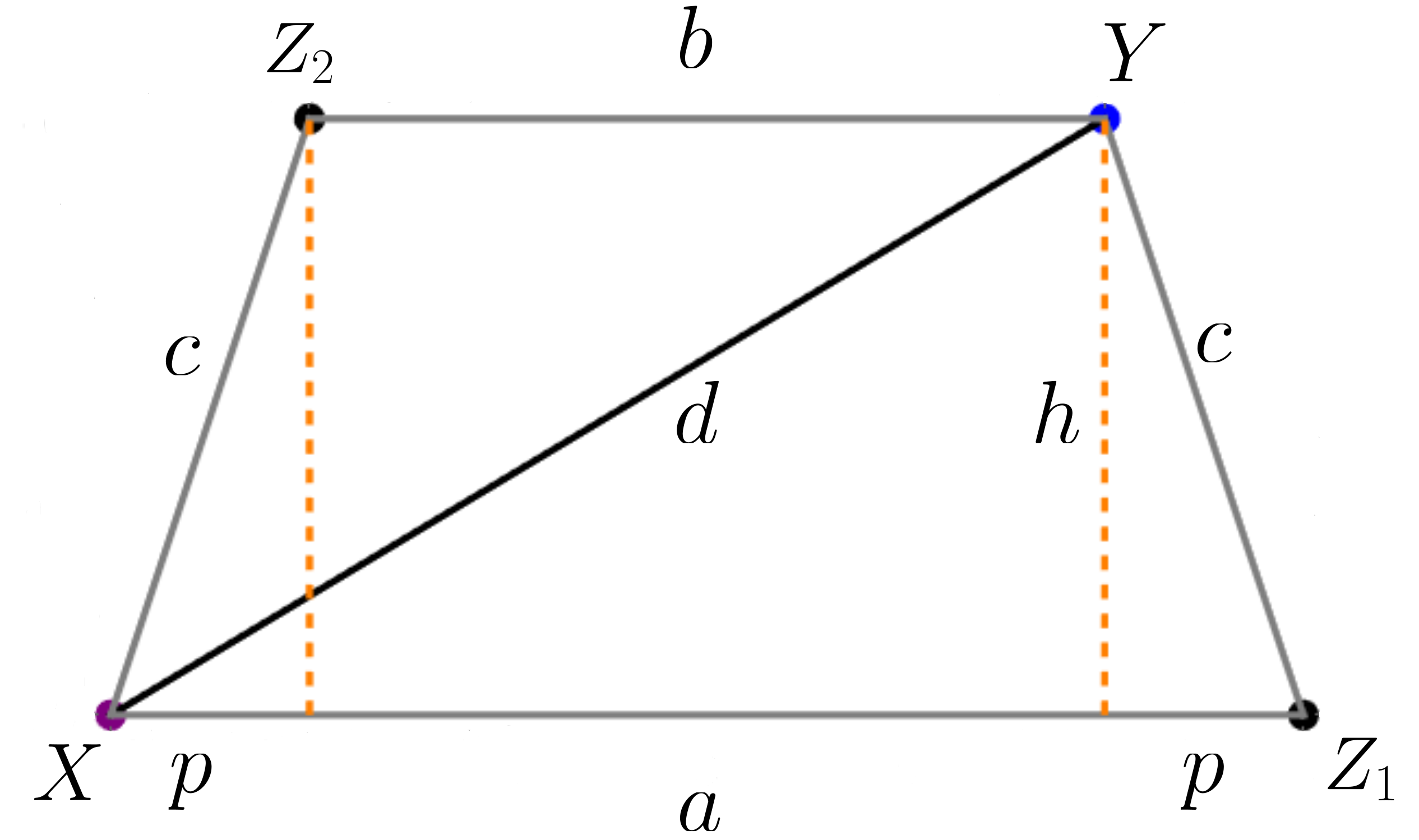}}			
			\caption{ 
				(a) The torus with the four points $X,Y,Z_1,Z_2$ with the parameters $R=2,r=1,u=\frac{11\pi}{6},v=\frac{\pi}{3},\theta=\frac{\pi}{3},\varphi=\frac{5\pi}{6}$. In red the two arcs of meridians considered in ${d^*}^2$, in yellow the parallel at angle $\varphi$ and in brown the parallel at angle $\theta$. The four points belong to the same plane $\alpha: 1.43x+0.38y+3.92z-4.96=0$ plotted in light blue. (b) The construction of the isosceles trapezoid in the plane $\alpha$; we use Pythagoras' theorem to find an equality involving the two basis $a$ and $b$, the edges $c$ and the diagonal $d$.}
			\label{trapezi}
		\end{figure}\\
		\noindent  We prove that the functional is bounded from below.
		Taken any two points $X$ and $Y$ on $\operatorname{bd} T_r$ and taken $Z_1$ and $Z_2$ as in \eqref{eq:z1z2}, we know that all the parallels on a torus are planar circular arcs, that lie in plane parallels to the plane $z=0$ of the centerline, while the meridians are all congruent circles. As shown in Figure \ref{trapeziodasoloo}, we consider the segments $XZ_1$ with length $a$, $YZ_2$ with length $b$, that have in general different lengths and the congruent segments $XZ_2$  and $YZ_1$ with length $c$, while we denote with $\overline{XY}=d$ the length of the segment that connects $X$ and $Y$.\\
		As shown in Figure \ref{trapeziotoro}, the four points all lie on a plane $\alpha$, and thus they define a planar polygon with four edges. Indeed, simple calculations show immediately that 
		$$
		\det \begin{pmatrix}
			XZ_2\\
			XZ_1\\
			XY
		\end{pmatrix}= \det\begin{vmatrix}
			r \cos{u}(\cos{\varphi}-\cos{\theta}) & r \sin{u}(\cos{\varphi}-\cos{\theta}) & c\\
			a(\cos{u}-\cos{v}) & a(\sin{u}-\sin{v}) & 0\\
			a\cos{u}-b\cos{v} & a\sin{u}-b\sin{v} &c
		\end{vmatrix}=0
		$$
		i.e. the four points lie on the same plane.\\
		Moreover, $YZ_2=(b(\cos{u}-\cos{v}), b(\sin{u}-\sin{v}),0)$, that is clearly parallel to $XZ_1$ and this implies that the quadrilateral defined in the plane $\alpha$ is an isosceles trapezoid.
		Now, from a standard reasoning, using Pythagoras' Theorem and calling $p=(a-b)/2$ and $h$ the height, we have $$
		|X-Y|^2=d^2=h^2+(a-p)^2=c^2-p^2+a^2+p^2-2ap=c^2+ab=(XZ_2)^2+XZ_1\cdot YZ_2
		$$
		Now, it is obvious that the segment $XZ_2$ is shorter than the shortest arc of meridian connecting $X$ and $Z_2$ (and the same for $Y$ and $Z_1$) and that $XZ_1$ and $YZ_2$ are shorter than the minimal arc of parallel connecting the points (red, yellow and brown arcs in Figure \ref{trapeziotoro}), so it is true that $$|X-Y|^2\leq l_M^2+\hat{l}_{P(\theta)}\hat{l}_{P(\varphi)}$$
		and thanks to Remark \ref{planar} the second term is exactly ${d^*}^2(X,Y)$.\\
		Then ${d^*}^2(X,Y)\geq |X-Y|^2$, and $$\frac{1}{|X-Y|^2}\geq \frac{1}{{d^*}^2(X,Y)}$$
		and the integrand is greater or equal to $0$ and is integrated on bounded intervals; this implies that the functional is positive and thus bounded from below.
	\end{proof}
	\begin{remark}
		\label{annullamento}
		\emph{We have that ${d^*}^2(X,Y)=0$ holds iff $s=t$ and $\theta=\varphi$, i.e. $X=Y$.\\
		While one implication is trivial, if ${d^*}^2(X,Y)$ vanishes, then, being the sum of positive quantities, both must vanish, so $l_M=0$ implies $\theta=\varphi$ and
		$$\int_s^tr|\tau(\xi)|\,d\xi=0$$
		which implies either $t=s$, i.e. $X=Y$, or $\tau(\xi)=0$ for every $\xi$, but then
		the arc of centerline is a plane curve. Now from $\hat l_{P(\theta)}\hat l_{P(\varphi)}=0$ and $s\neq t$ it follows $1-r \cos \theta \kappa(\xi)=0$ for all $\xi$ as above, or $1-r \cos \varphi \kappa(\xi)=0$. But then in both cases $\kappa$ is constant and therefore the above arc is a circumference with radius $\rho=r\cos\theta$ and $\kappa\rho=1$, which implies local interpenetration. Therefore it must be $s=t$ and thus $X=Y$.}
	\end{remark}
	We now prove a boundedness result for a general regular closed midline.
	\begin{theorem}
		\label{order}
		The functional is bounded from below, i.e. $F(T_r)>-\infty$, for any $T_r[\gamma]$ with $\gamma\in C^3([0,L],\mathbb{R}^3)$.
	\end{theorem}
	\begin{proof} Since integration is on bounded intervals, the boundedness from below of the integrand is sufficient. The integrand is given by the difference of two positive quantities, so it can degenerate to $-\infty$ only if the denominator of the second term goes to $0$ while the denominator of the first one is bounded away from zero or it converges to $0$ with a different order. From the previous Remark, ${d^*}^2$ tends to $0$ iff $X\to Y$ on the rod, i.e. $s\to t$ and $\theta \to \varphi$, then also $|X-Y|\to 0$. Both fractions diverge and we have to carefully study the behaviour of the integrand.\\
		We can assume $X(s,\theta)$, $Y(t,\varphi)$, with $t=s+\eta_1$ and $\varphi=\theta+\eta_2$
		where $\eta_i$ are respectively a small displacement along the centerline and a small variation in the angle on the cross-section.
		We also assume a positive torsion; with a negative torsion of $\gamma$ the proof is analogous. We expand the expression of $|X-Y|$ up to the fourth order and after long but straightforward calculations we find from \eqref{eq:Pdsth} and the orthogonality of $\boldsymbol{t},\boldsymbol{n},\boldsymbol{b}$
		$$|X-Y|^2=|p(r,s,\theta)-p(r,s+\eta_1,\theta+\eta_2)|^2=A_2+A_3+A_4+o_5(\eta_1,\eta_2)$$
		where $A_k$ denote the terms of $k$th degree in $\eta_1,\eta_2$ and $o_5$ is at least of fifth degree in $\eta_1,\eta_2$. We notice that
		$$\begin{aligned}
			A_2&=\eta_1^2+r^2\cos^2\theta\eta_1^2\kappa(s)^2-2\eta_1^2r\cos\theta \kappa(s)+r^2\eta_2^2+r^2\tau(s)^2\eta_1^2+2r^2\eta_1\eta_2\tau(s)=\\
			&=(\eta_1^2(1-r \kappa(s) \cos \theta)^2+r^2(\eta_2+\eta_1 \tau(s))^2.
		\end{aligned}
		$$
		The expressions of $A_k$ for $k=3,4$ are given in the Appendix.
		On the other hand,
		$$
		{d^*}^2(X,Y)=\int_{s}^{s+\eta_1}\!\!\!\!\!\!\!\!\!\!\!\! \left|1-r \cos \theta \kappa(\xi)\right|\, d\xi\int_{s}^{s+\eta_1}\!\!\!\!\!\!\!\!\!\!\! \left|1-r \cos (\theta +\eta_2)\kappa(\xi)\right|\, d\xi+r^2\left(\eta_2+\int_s^{s+\eta_1}\!\!\!\!\!\!\!\!\!\!\!\!\tau(\xi)\,d\xi\right)^2$$
		and expanding this up to the fourth order we find
		$${d^*}^2(X,Y)=A_2+A_3+B_4+o_5(\eta_1,\eta_2)$$
		where $A_k,B_k$ denote the terms of $k$th degree in $\eta_1,\eta_2$ and $o_5$ is at least of fifth degree in $\eta_1,\eta_2$. 
		While $A_2,A_3$ coincide with the above ones, $B_4$ is different and its expression is given in the Appendix. 
		Then, the integrand becomes $$\frac{1}{A_2+A_3+A_4+o_5(\eta_1,\eta_2^b)}-\frac{1}{A_2+A_3+B_4+o_5(\eta_1,\eta_2)}$$
		which behaves for $\eta_1,\eta_2\to0$ as
		$$\frac{B_4-A_4+o_5}{A_2^2+o_5}$$
		and which can diverge in $(\eta_1,\eta_2)=(0,0)$ if and only if $A_2$ tends to zero. But it is easily seen that $A_2$ is positive definite if $1-r \kappa(s) \cos \theta\neq 0$, which means that $r\kappa(s)<1$ which is true by assumption of local non-interpenetration of matter. Then the integrand is also bounded, since it is bounded when $\eta_1,\eta_2\to0$ and elsewhere bounded.
	\end{proof}  
	\begin{theorem}
		\label{divergence}
		The functional degenerates to $+\infty$ if and only if there is at least a point of self-contact between points belonging to $\operatorname{bd} T_r$, without interpenetration of matter.
	\end{theorem}
	\begin{proof}
		It is sufficient to consider the case of a single isolated self-contact point, so we assume that there exist $A$, $B$ in $\operatorname{bd} T_r$ such that $|A-B|=0$, with $A(s,\theta)\neq B(t,\varphi)$ (that implies $s\neq t$). As pointed out in \cite{chamekh2009modeling,chamekh2020frictional,gonzalez1999global,gonzalez2002global,mlika2018nitsche,schuricht2004characterization}, without losing in generality we can consider two disjoint thin cylinders with perpendicular axes touching externally. Thus we take $X=(x,r\sin\theta,r-r\cos\theta)$ and $Y=(-r\sin\varphi,y,-r+r\cos\varphi)$, $\theta,\varphi\in [0,2\pi]$ and $x,y$ belonging to a suitable small interval $[-L,L]$, $L>0$.
		Now
		$$|X-Y|^2=(x+r\sin\varphi)^2+(y-r\sin\theta)^2+(r(1-\cos\theta)+r(1-\cos\varphi))^2$$
		Since $X\neq Y$ elsewhere, we consider only $$I=\int\limits_{-L}^L\int\limits_{-L}^L\int\limits_{-\pi}^{\pi}\int\limits_{-\pi}^{\pi}\frac{1}{|X-Y|^2}\,d\varphi\,d\theta\,dx\,dy.$$ 
		Using the inequality $(x\pm y)^2\leq 2(x^2+y^2)$, we obtain $$(x+r\sin\varphi)^2\leq 2(x^2+r^2\sin^2\varphi),\quad (y-r\sin\theta)^2\leq 2(y^2+r^2\sin^2\theta)$$
		and
		$$(r(1-\cos\theta)-r(1-\cos\varphi))^2\leq 2r^2((1-\cos\theta)^2+(1-\cos\varphi)^2).$$ 
		Replacing into $I$ and using standard trigonometric identities, we get
		$$I\geq\frac{1}{2}\int\limits_{-L}^L\int\limits_{-L}^L\int\limits_{-\pi}^{\pi}\int\limits_{-\pi}^{\pi}\frac1{x^2+y^2+r^2(1-\cos\theta)+r^2(1-\cos\varphi)}\,dx\,dy\,d\varphi\,d\theta.$$ 
		Since the integrand is positive, integration on $[-L,L]^2$ is greater than on $C=\{(x,y):x^2+y^2\leq L^2\}$. So 
		$$I\geq\frac{1}{2}\int_C\int\limits_{-\pi}^{\pi}\int\limits_{-\pi}^{\pi}\frac1{x^2+y^2+r^2A^2}
		\,d\varphi\,d\theta\,dx\,dy=:J,\quad \text{with}\,\,\, A^2=(1-\cos\theta)+(1-\cos\varphi).$$
		$J$ is easily computed with polar coordinates and it yields
		$$J=\frac{\pi}{2}\int\limits_{-\pi}^{\pi}\int\limits_{-\pi}^{\pi}\log\left(1+\frac{L^2}{r^2A^2}\right)\,d\theta\,d\varphi$$
		and since the length $L$ can be assumed small, by standard estimates we get
		$$I\geq\frac{C\pi}{2}\frac{L^2}{r^2}\int\limits_{-\pi}^{\pi}\int\limits_{-\pi}^{\pi}\frac{2}{\theta^2+\varphi^2}\,d\theta\,d\varphi$$
		which is positively divergent.
		\\
		Viceversa, if the functional is positively divergent, then this must be caused by its first integrand, so $|X-Y|=0$ and $X$ and $Y$ occupy the same position in space. Thus, it may be ${d^*}^2(X,Y)=0$ or ${d^*}^2(X,Y)\neq0$. In the first case, $X$ and $Y$ coincide also in $\Omega$ (i.e. have the same coordinates) and by Theorem \ref{order}, the functional is convergent. In the second case $X$ and $Y$ do not coincide in $\Omega$ and we have self-contact.
	\end{proof}
	By Theorem	\ref{divergence} and Corollary \ref{inj} it follows that the functional diverges also for neighbourhoods with interpenetration of matter.\\
	
	We denote by $\mathbb{T}$ the tubular neighbourhood of a circle $S^1$ in $\mathbb{R}^3$ of length $L$ and we define $\mathcal{L}=\{g: \mathbb{T}\to T_r[\gamma]\subset\mathbb{R}^3\}$.
	\begin{theorem}
		The functional $F:\mathcal{L}\to \mathbb{R}$ is continuous with respect to the $C^3$-topology of the centerline $\gamma$.  
	\end{theorem}
	\begin{proof}
		We want to prove that small deformations of $\gamma$ 
		produce small variations in the functional. For a given positive number $\varepsilon$ we have to show that there exist four positive numbers $d_k$, $k=0,\dots,3$ such that if $\eta$ is an homeomorphism of class $C^3$ with $|\eta'(t)|=1,\, \forall t$ such that
		$$|\gamma^{(k)}(t)-\eta^{(k)}(t)|<d_k\qquad\hbox{for all $t$ and for $k=0,\ldots,3$}$$
		then we have $|F(\gamma)-F(\eta)|<\varepsilon$.\\
		The Euclidean distance and ${d^*}^2$ are compositions of continuous functions of $\gamma$ and the Serret-Frenet frame, the curvature and the torsion of $\gamma$, which are continuous with respect to the $C^3$-topology of the centerline $\gamma$.\\
		Then, by the continuity of the function $x\mapsto 1/x$, $\forall \varepsilon>0$ there exists a neighbourhood $I\subseteq C^3([0,L], \mathbb{E}^3)$ such that for every $\eta\in I$ and $u=(s,t,\theta,\varphi)$,
		$$
		\left| \frac{1}{|X-Y|_\eta^2(u)}-\frac{1}{|X-Y|_\gamma^2(u)}  \right|<\frac{\varepsilon}{8\pi^2L^2},\,\, \left| \frac{1}{{d^*}^2(X,Y)_\eta(u)}-\frac{1}{{d^*}^2(X,Y)_\gamma(u)}  \right|<\frac{\varepsilon}{8\pi^2L^2}
		$$
		whence it follows 
		$$
		\left|\left(\frac{1}{|X-Y|_\eta^2(u)}-\frac{1}{{d^*}^2(X,Y)_\eta(u)}\right)-\left(\frac{1}{|X-Y|_\gamma^2(u)}-\frac{1}{{d^*}^2(X,Y)_\gamma(u)}\right)\right|<\frac{\varepsilon}{4\pi^2L^2}
		$$
		and integrating, we have $|F(\eta)-F(\gamma)|<\varepsilon$.
	\end{proof}
	\section{M\"obius energy as a particular case}
	In \cite{freedman1994mobius}, Freedman and al. proved that M\"obius energy presents a minimizer in each isotopy class, a so called ``optimal knot'', and that the circle is the global minimizer between all these classes.\\
	We show now that letting $r=0$ we obtain the particular case of the O'Hara energy functional \eqref{hara} for knots. We limit our treatment here to simple considerations when $r=0$, but a further study could include results on $\Gamma$-convergence of our functional to O'Hara one when $r\to 0$.\\
	If we assume that the radius of the neighbourhood vanishes, we have the following modifications.
	\begin{enumerate}
		\item the tubular neighbourhoods are contracted to knots described by the curve $\gamma$, that can be now taken in $C^2([0,L],\mathbb{R}^3)$ as proposed by O'Hara in \cite{o1991energy};
		\item the notion of interpenetration of matter loses meaning: indeed, a knot does not have internal points. The space of knots is subdivided into cells, where the region plotted in orange in Figure \ref{isotopy} does not exist anymore. The boundary of the green region in each cell, i.e. the configuration of self-contact of rods is expanded and coincides with  the boundary of the cell itself: it indicates the coincidence between points on the knot, i.e. a singular knot with at least a double point, as in O'Hara's treatment;
		\item the functional reduces to \eqref{hara}, i.e. to O'Hara's one; the Euclidean distance between points on $\operatorname{bd} T_r$ becomes the Euclidean distance between points on the knot and ${d^*}^2$ becomes the squared difference between the curvilinear abscissas that describe the position of $X$ and $Y$.
	\end{enumerate}
	\section{Discussion on the exponent}
	Many generalizations of M\"obius energy were studied in the last two decades, involving in particular the use of a different exponent for the two denominators of the functional or for the integrand, and the minimization of the energy or its coupling with an elastic energy on the knot, see for details \cite{o1992family,kim1993torus,freedman1994mobius,o1994energy,buck1995simple,simon1996energy,von1998minimizing,o2008energy,blatt2013stationary}. We carry out here a brief discussion on the possible exponent that can be chosen inside the integrand, recovering the same result obtained for the M\"obius energy.\\
	In the definition of the energy we used an exponent $2$ for the Euclidean distance, and the same was intuitively done also for $d^*$. We can now study which are the exponents that guarantee divergence of the functional only in presence of a self-contact point: in particular, for O'Hara's energy in \eqref{hara}, as exponent one can choose $\alpha\in [2,3)$: $\alpha<3$ provides a compensation of the two fractions in the integrand when $X\to Y$, while $\alpha\geq 2$ assures divergence for a double point.\\
	The same result is valid also for tubular neighbourhoods.
	\begin{theorem}
		The energy functional \begin{equation}
			\label{funcexp}	F(T_r)=\int_{\operatorname{bd} T_r \times \operatorname{bd} T_r}\left(\frac{1}{|X-Y|^\alpha}-\frac{1}{({d^*}^2(X,Y))^{\frac{\alpha}{2}}}\right)dS\,dS
		\end{equation} is bounded on tubular neighbourhoods without self-contact or interpenetration of matter if $\alpha\in (0,3)$, $\alpha\in \mathbb{R}$. Moreover, the functional diverges for a finite number of self-contact points on $\operatorname{bd} T_r$ if $\alpha\geq 2$ and converges for $\alpha<2$. Finally, if there is a line of self-contact points, \eqref{funcexp} diverges also for $\alpha\in [1,2)$.
	\end{theorem}
	\begin{proof}
		As explained above in Theorem \ref{order}, the energy is in general bounded for such objects without self-contact and the unique case that one has to carefully study is when $(t,\varphi)\to (s,\theta)$, i.e. we have to prove that $|X-Y|^\alpha \to 0$ and ${d^*}^2(X,Y)^{\alpha/2} \to 0$ with the same order. We have, using the same notation of the proof of Theorem \ref{order},
		$$
		F^\alpha(T_r)=\int_{\operatorname{bd} T_r}\int_{\operatorname{bd} T_r} \frac{1}{(|X-Y|^2)^\frac{\alpha}{2}}-\frac{1}{{d^*}^2(X,Y)^{\frac{\alpha}2}}
		$$
		$$=\int_{\operatorname{bd} T_r}\int_{\operatorname{bd} T_r} \frac{1}{(A_2+A_3+A_4+o_5(\eta_1,\eta_2)_{|\cdot |})^{\frac\alpha2}}-\frac{1}{(A_2+A_3+B_4+o_5(\eta_1,\eta_2)_{{d^*}^2}))^{\frac{\alpha}{2}}}=
		$$
		$$=\int_{\operatorname{bd} T_r}\int_{\operatorname{bd} T_r}\frac{(A_2+A_3+B_4+o_5(\eta_1,\eta_2)_{{d^*}^2})^\frac{\alpha}{2}-(A_2+A_3+A_4+o_5(\eta_1,\eta_2)_{|\cdot |})^\frac{\alpha}{2}}{[(A_2+A_3+A_4+o_5(\eta_1,\eta_2)_{|\cdot |})(A_2+A_3+B_4+o_5(\eta_1,\eta_2)_{{d^*}^2})]^{\frac{\alpha}{2}}} 
		$$
		We study the integrand in a neighbourhood of $(\eta_1,\eta_2)=(0,0)$, then having all bounded quantities, in both the numerator and the denominator we retain the polynomials with lower degree: the denominator of the integrand is a polynomial with the smallest degree equal to $4\frac{\alpha}{2}=2\alpha$, while with an expansion of the numerator, we have a polynomial with smallest degree equal to $2+\alpha$.\\
		Then, to have a convergent integral we must require $2+\alpha+1> 2\alpha$, then for $\alpha < 3$ the functional is finite for tubular neighbourhoods without self-contact or interpenetration of matter.\\
		For the proof of the divergence in case of at least a self-contact point with $\alpha\geq 2$, we proceed as in Theorem \ref{divergence}. Suppose there exists an isolated point of self-contact. We study the behaviour of the integral
		$$
		I_\alpha=\int\limits_{-L}^L\int\limits_{-L}^L\int\limits_{-\pi}^\pi\int\limits_{-\pi}^\pi\frac{1}{|X-Y|^\alpha}\,d\varphi\,d\theta\,dx\,dy=:\int\limits_{[-L,L]^2\times[-\pi.\pi]^2}\frac{1}{D^\alpha}\,d\varphi\,d\theta\,dx\,dy
		$$
		Notice first that if $\Omega_L=[-L,L]^2\times[-\pi.\pi]^2$ and $\Omega_\varepsilon=[-\rho\varepsilon,\rho\varepsilon]^2\times[-\varepsilon,\varepsilon]^2$, then
		$$I_\alpha=\int_{\Omega_\varepsilon}\frac{1}{D^\alpha}+\int_{\Omega_L\setminus\Omega_\varepsilon}\frac{1}{D^\alpha}=:J^\varepsilon_\alpha+k_\varepsilon\geq J^\varepsilon_\alpha$$
		with $0<k_\varepsilon<+\infty$ since $D^\alpha$ is bounded away from $\Omega_\varepsilon$. Therefore $J^\alpha_\varepsilon\leq I_\alpha=J^\varepsilon_\alpha+k_\varepsilon$ and then $I_\alpha$ is divergent iff $J^\varepsilon_\alpha$ is divergent.
		Now, as before,
		$$|X-Y|^2=(x+r\sin\varphi)^2+(y-r\sin\theta)^2+(r(1-\cos\theta)+r(1-\cos\varphi))^2$$
		and setting $\xi=x/r$, $\eta=y/r$, $J^\varepsilon_\alpha$ reduces to 
		\begin{equation}
			\label{eq:Jalfa}
			J^\varepsilon_\alpha=\frac{1}{r^\alpha}\!\!\!\!\int\limits_{[-\varepsilon,\varepsilon]^4}\!\!\!\!\frac{1}{\left((\xi+\sin\varphi)^2+(\eta-\sin\theta)^2+((1-\cos\theta)+(1-\cos\varphi))^2\right)^\frac{\alpha}{2}}\,d\varphi\,d\theta\,dx\,dy 
		\end{equation}
		We study now when $J^\varepsilon_\alpha$ diverges. We have 
		$$
		(2-\cos \theta-\cos \varphi)^2=2\left(\sin^2\frac{\theta}{2}+\sin^2\frac{\varphi}{2}\right)^2\leq 2\left(\frac{\theta}{2}+\frac{\varphi}{2}\right)^2\leq 4\left(\frac{\theta^2}{4}+\frac{\varphi^2}{4}\right)\leq \theta^2+\varphi^2
		$$
		By the change of coordinates
		$$
		u=\xi+\sin \varphi,\quad v=\eta-\sin \theta, \quad z=\varphi,\quad t=\theta
		$$
		with Jacobian determinant equal to $1$, $J^\varepsilon_\alpha$ satisfies
		$$J^\varepsilon_\alpha\geq\int\frac{1}{(u^2+v^2+z^2+t^2)^{\frac{\alpha}{2}}}\,dz\,dt\,du\,dv
		$$
		with $u,v\in [-\varepsilon-\sin \varepsilon,\varepsilon+\sin \varepsilon]\supseteq[-\varepsilon,\varepsilon]$, $z,t \in [\varepsilon,\varepsilon]$. Thus,
		\begin{equation}\label{eq:divJalpha}
			J^\varepsilon_\alpha\geq \iint\limits_{[-\varepsilon,\varepsilon]^2}\iint_{C_\varepsilon}\frac{1}{(A^2+z^2+t^2)^{\frac{\alpha}{2}}}\,dz\,dt\,du\,dv
		\end{equation}
		with $A^2=u^2+v^2$ and $C_\varepsilon=\{(z,t):z^2+t^2\leq \varepsilon^2\}$. The integral on $C_\varepsilon$ is easy to compute and yields for $\alpha\neq2$ (the case $\alpha=2$ was treated in Theorem \ref{divergence})
		$$\frac{2\pi}{(2-\alpha)}A^{2-\alpha}\left(\left(1+\frac{\varepsilon^2}{A^2}\right)^\frac{2-\alpha}{2}-1\right)\geq c_\varepsilon \pi \varepsilon^2\frac{1}{A^\alpha}=c_\varepsilon\frac{\pi \varepsilon^2}{(u^2+v^2)^\frac{\alpha}{2}}
		$$
		for $\varepsilon$ sufficiently small (but finite) and $c_\varepsilon>0$.
		Thus,
		\begin{equation}\label{eq:divJalpha2}
			J_\alpha^\varepsilon\geq c_\varepsilon\pi\varepsilon^2\int\limits_{-\varepsilon}^\varepsilon\int\limits_{-\varepsilon}^\varepsilon \frac{1}{(u^2+v^2)^\frac{\alpha}{2}}\,du\,dv  
		\end{equation}
		which is easily seen to be divergent if $\alpha-1\geq 1$, i.e. $\alpha\geq 2$.
		
		To discuss convergence, we neglect the term $(2-\cos\theta-\cos\varphi)^2$ in the denominator in \eqref{eq:Jalfa} and we have
		$$J^\varepsilon_\alpha\leq\frac{1}{r^\alpha}\int_{[-\varepsilon,\varepsilon]^4}\frac{1}{\left((\xi+\sin\varphi)^2+(\eta-\sin\theta)^2\right)^\frac{\alpha}{2}}\,d\varphi\,d\theta\,dx\,dy$$
		and by the same coordinate change as above this turns into
		$$J^\varepsilon_\alpha\leq\frac{1}{r^\alpha}\iint\limits_{[-\varepsilon,\varepsilon]^2}\iint_{C_{\sqrt2\varepsilon}}\frac{1}{(u^2+v^2)^{\frac{\alpha}{2}}}\,du\,dv\,dz\,dt\leq\frac{8\pi\varepsilon^3}{r^\alpha}\int_0^{\sqrt 2\varepsilon}\frac{1}{\sigma^{\alpha-1}}\,d\sigma$$
		by an analogous reasoning as above. So, $J^\varepsilon_\alpha$ converges (if there is only one point of self-contact) if $\alpha<2$.
		
		It remains to investigate when the integral diverges if there is a line of points of contact on $\operatorname{bd}T_r$. Without loss of generality, we can suppose to have contact between two parallel cylinders of length $L$ along the $x$-axis. With analogous notation as above, we have $X(x_1,r\sin\theta,r(1-\cos\theta))$ and $Y(x_2,r\sin\varphi,r(1-\cos\varphi))$ and
		$$\|X-Y\|^2=(x_1-x_2)^2+r^2(\sin\theta-\sin\varphi)^2+r^2(2-\cos\theta-\cos\varphi)^2.$$
		Now it is no more possible to restrict to $[-\varepsilon,\varepsilon]^2$ on the $(x_1,x_2)$ coordinates, but still on the angles. Then $I_\alpha$ will diverge if 
		$$J^\varepsilon_\alpha=\iint\limits_{[-L,L]^2}\iint\limits_{[-\varepsilon,\varepsilon]^2}\frac{1}{[(x_1-x_2)^2+(\sin\theta-\sin\varphi)^2+(2-\cos\theta-\cos\varphi)^2]^{\alpha/2}}\,dx_1\,dx_2\,d\theta\,d\varphi$$
		diverges. With the same reasoning as above, it is easy to see that if $\varepsilon$ is small enough, then
		$$J^\varepsilon_\alpha\geq\iint_{[-L,L]^2}\iint_{[-\varepsilon,\varepsilon]^2}\frac{C}{(B^2+\theta^2+\varphi^2)^{\alpha/2}}\,dx_1\,dx_2\,d\theta\,d\varphi$$
		where $B^2=(x_1-x_2)^2$ and $C$ is a positive constant depending on $\varepsilon$. The right-hand side is analogous to the right-hand side of \eqref{eq:divJalpha} and therefore, as in \eqref{eq:divJalpha2}
		$$
		J_\alpha^\varepsilon\geq D\int\limits_{-L}^L\int\limits_{-L}^L \frac{1}{(x_1-x_2)^\alpha}\,dx_1\,dx_2
		$$
		fo a suitable positive constant $D$, including in this case also $\alpha=2$. This is  easily seen to be divergent if $\alpha\geq 1$ by switching to the coordinates $\xi=x_1-x_2,\eta=x_1+x_2$. This completes the proof.
	\end{proof}
	\begin{remark}
		\emph{The above results show that if the total energy (for generic $\alpha<3$) diverges, then this must be ascribed to the divergence of the term with the Euclidean distance if $\alpha\geq 2$ and there is at least a point of self-contact. However, in this case the divergence of the integral cannot distinguish between a single (or finite) point of contact and a line of contact. On the other side, if $1\leq \alpha<2$ and the integral diverges, then there must be a line of points of contact.} 
	\end{remark}
	\begin{remark}
		\emph{The result is the same as in \cite{freedman1994mobius} for a finite number of self-contact intersections: the admissible range for $\alpha$ is $[2,3)$, that provides a finite energy for knots without double points and diverges in the case of loss of injectivity.}
	\end{remark}
	\begin{remark}
		\emph{In the setting of knots, it is easy to see that in O'Hara's functional in \eqref{hara} with general exponent $\alpha$, for nearby points on a circle, one is faced with the limit
		$$
		\lim_{x\to 0}\frac{1}{2^\alpha R^\alpha\left|\sin\left(\frac{x}{2}\right)\right|^\alpha}-\frac{1}{R^\alpha x^\alpha}
		$$ 
		that is finite for $\alpha\in (0,3)$. \\
		One could do the same  on the surface of a torus with exponents $\alpha=1,3,4$ (the general case is far too complicated). After long calculations that involve the expansions of incomplete and complete elliptic integrals of the first, second and third kind the limit turns out to be bounded for $\alpha=1$ and unbounded in the other two cases, as expected.}
	\end{remark}
	\section{Final remarks}
	The object of a forthcoming work could be to prove, similarly to what was done with M\"obius energy, but using the Direct Method of Calculus of Variations, that in each isotopy class the functional is minimal on the neighbourhood of an ideal optimal centerline, while the torus with circular midline (analogous of a circular knot for M\"obius energy) is the minimum between all the possible configurations. Such a result could be carried out expressing the functional in terms of the central curve and its Serret-Frenet frame, $(\gamma,(\boldsymbol{t},\boldsymbol{n},\boldsymbol{b}))\in W^{1,2}([0,L],\mathbb{R}^3)$, using then Calculus of Variations and Schuricht's results instead of tools of knot theory and M\"obius transformations. However, the research for a necessary condition to prove the coercivity of the functional is quite difficult.
	\begin{appendices}
	\section{Computations for the proof of Theorem \ref{order}}
	$$\scalemath{1}{\begin{aligned}
			A_3&=\eta_1^2 (\eta_2 r^2 \tau'(s) \cos^2\theta + \eta_2 \kappa(s) r \sin\theta - 
			\eta_2 \kappa(s)^2 r^2 \cos\theta \sin\theta + 
			\eta_2 r^2 \tau'(s) \sin^2\theta)+\\
			&+\eta_1^3 (-\kappa'(s) r \cos\theta + \kappa(s) \kappa' (s)r^2 \cos^2\theta + 
			r^2 \tau'(s) \tau(s) \cos^2\theta + r^2 \tau'(s)\tau(s)\sin^2\theta);\\
			A_4&=\eta_1^4 \left(\frac{1}{3} \kappa''(s) \kappa(s) r^2 \cos ^2\theta -\frac{1}{3} \kappa''(s) r \cos \theta -\frac{1}{12} \kappa(s)^4 r^2 \cos ^2\theta +\frac{1}{6} \kappa(s)^3 r \cos \theta\right.-\frac{1}{12} \kappa(s)^2 r^2 \tau(s)^2 \sin^2\theta+ \\
			&-\frac{1}{6} \kappa(s)^2 r^2 \tau(s)^2 \cos ^2\theta -\frac{1}{6} \kappa(s)^2 r^2 \tau'(s) \sin \theta  \cos \theta -\frac{\kappa(s)^2}{12}+\frac{1}{6} \kappa(s) \kappa'(s) r^2 \tau(s)  \sin \theta  \cos \theta+ \\
			&+\frac{1}{6} \kappa(s) r \tau'(s) \sin \theta 
			+\frac{1}{6} \kappa(s) r \tau(s)^2 \cos \theta +\left.\frac{1}{4} \kappa'(s)^2 r^2 \cos ^2\theta+\frac{1}{3} r^2 \tau''(s) \tau(s) \sin ^2\theta -\frac{r^2 \tau(s)^4}{12}+\frac{r^2 \tau'(s)^2}{4}\right)+\\
			&+\eta_1^3 \left(-\frac{1}{3} \eta_2 \kappa(s)^2 r^2 \tau(s) -\eta_2 \kappa(s) \kappa'(s) r^2 \sin \theta  \cos \theta +\frac{1}{3} \eta_2 \kappa(s) r \tau(s)  \cos \theta+\frac{2}{3} \eta_2 \kappa'(s) r \sin \theta -\frac{1}{3} \eta_2 r^2 \tau(s)^3  +\right.\\
			&\left.+\frac{1}{3} \eta_2 r^2 \tau''(s) \sin ^2\theta\right)+\eta_1^2 \left(-\frac{1}{2} \eta_2^2 \kappa(s)^2 r^2 \cos ^2\theta +\frac{1}{2} \eta_2^2 \kappa(s) r \cos \theta -\frac{1}{2} \eta_2^2 r^2 \tau (s)^2\right)+\frac{1}{3} \eta_1 \eta_2^3 r^2 \tau(s)-\frac{\eta_2^4 r^2}{12};
	\end{aligned}}
	$$
	$$\scalemath{1}{\begin{aligned}
			B_4&=\eta_1^4 \left(\frac{1}{3} \kappa''(s) \kappa(s) r^2 \cos ^2\theta-\frac{1}{3} \tau''(s) r \cos \theta+\frac{1}{4} \kappa'(s)^2 r^2 \cos ^2\theta+\frac{1}{3} r^2 \tau''(s) \tau(s) +\frac{r^2 \tau'(s)^2}{4}\right)\\
			&+\eta_1^3 \left(-\eta_2 \kappa(s) \kappa'(s) r^2 \sin \theta) \cos \theta+\frac{1}{2} \eta_2 \kappa'(s) r \sin \theta+\frac{1}{3} \eta_2 r^2 \tau'(s)\right)+\eta_1^2 \left(\frac{1}{2} \eta_2^2 \kappa(s) r \cos \theta-\frac{1}{2} \eta_2^2 \kappa(s)^2 r^2 \cos ^2\theta\right).
	\end{aligned}}$$
\end{appendices}
	\section*{Acknowledgements}
	CL and AM are supported by Gruppo Nazionale per la Fisica Matematica (GNFM) of Istituto Nazionale per l'Alta Matematica (INdAM).
	\nocite{*} 
	\bibliographystyle{acm}

\end{document}